\documentclass[journal]{IEEEtran}
\usepackage{amsmath,amsfonts,amssymb, amsthm}
\usepackage{algorithmic}
\usepackage{stfloats}
\usepackage{algorithm}
\usepackage{graphicx}
\usepackage{siunitx}
\usepackage{setspace}
\usepackage{array}
\usepackage{orcidlink}
\usepackage{hyperref}
\hypersetup{
	colorlinks=true,
	linkcolor=black,
        citecolor=black,
 }
 
\theoremstyle{remark}
\newtheorem{theorem}{Theorem}
\newtheorem{assumption}{Assumption}

\usepackage{fancyhdr}
\usepackage{lipsum} 

\fancypagestyle{firstpage}{
    \fancyhf{}  
    \fancyhead[C]{\small This work has been submitted to the IEEE for possible publication.\\ 
    Copyright may be transferred without notice, after which this version may no longer be accessible.}
}

\begin{document}
\title{Decentralized Handover Parameter Optimization with MARL for Load Balancing in 5G Networks}
\author{Yang Shen~\orcidlink{0009-0000-0172-7955},
        Shuqi Chai~\orcidlink{0000-0002-8250-7782}, 
	Bing Li~\orcidlink{0000-0001-5226-7557},
        Xiaodong Luo,
        Qingjiang Shi~\orcidlink{0000-0003-0507-9080}, and
	Rongqing Zhang~\orcidlink{0000-0003-3774-6247}
	
\thanks{Yang Shen, Bing Li, and Rongqing Zhang are with the School of Computer Science and Technology, Tongji University, Shanghai 201804, China (e-mail: stanleyshen@tongji.edu.cn; lizi@tongji.edu.cn; rongqingz@tongji.edu.cn).}
\thanks{Shuqi Chai is with the Shenzhen Research Institute of Big Data, Shenzhen 518172, China (e-mail: schai@sribd.cn)}
\thanks{Xiaodong Luo is with the Chinese University of Hong Kong, Shenzhen, Shenzhen 518172, China (e-mail: xiaodongluo@cuhk.edu.cn)}
\thanks{Qingjiang Shi is with the School of Computer Science and Technology, Tongji University, Shanghai 201804, China, and also with the Shenzhen Research Institute of Big
Data, Shenzhen 518172, China (e-mail: shiqj@tongji.edu.cn)}
}

\maketitle
\thispagestyle{firstpage}
\begin{abstract}

In cellular networks, cell handover refers to the process where a device switches from one base station to another, and this mechanism is crucial for balancing the load among different cells. Traditionally, engineers would manually adjust parameters based on experience. However, the explosive growth in the number of cells has rendered manual tuning impractical.  Existing research tends to overlook critical engineering details in order to simplify handover problems. In this paper, we classify cell handover into three types, and jointly model their mutual influence. To achieve load balancing, we propose a multi-agent-reinforcement-learning (MARL)-based scheme to automatically optimize the parameters. To reduce the agent interaction costs, a distributed training is implemented based on consensus approximation of global average load, and it is shown that the approximation error is bounded. Experimental results show that our proposed scheme outperforms existing benchmarks in balancing load and improving network performance.

\end{abstract}

\section{Introduction}

\par 5G networks offer the capacity to connect many more devices compared to previous generations \cite{5G1,5G2}. Such a large number of connections pose significant challenges for base station (BS)  load balancing, primarily due to the uneven distribution and highly dynamic nature of user devices \cite{LB1,LB2}. In hotspot areas, cells experience congestion due to excessive demands, while the resources of idle cells remain underutilized. To address this issue, busy cells must offload some of their users to the surrounding idle cells through a process known as cell handover \cite{LB3}. Cell handover is essential for facilitating network load balancing and ensuring continuous connectivity. This process involves evaluating a set of decision conditions with adjustable threshold parameters to determine the required signal quality and the handover trigger mechanism \cite{3gpp-ts38104,3gpp-ts38133}. Traditionally, engineers relied on expert knowledge to manually adjust parameters.  However, in the current era of dense B5G/6G BS deployment, the complex interdependencies among handover parameters render manual adjustment methods both time-consuming and impractical.

\par To address this problem, current literature has explored various modeling and optimization methods for coordinating handover parameters. These approaches can be broadly categorized into heuristic-based and learning-based methods, each with distinct characteristics and applications. Heuristic-based approaches rely on predefined rules and empirical models to dynamically adjust handover parameters. In \cite{DDSO}, a genetic algorithm was proposed to optimize A5 parameters for inter-frequency handovers. The authors in \cite{HO3} proposed an adaptive handover parameter optimization method for mobility management, aiming to minimize unnecessary handovers such as ping-pong events and radio link failures. The study in \cite{HO2} introduced a feedback controller based self-organizing network (SON) algorithm for optimizing inter-frequency handovers, where parameters were adjusted depending on cell traffic and mobility conditions. In \cite{HO4}, a fuzzy logic controller combined with a weighted function was employed to self-optimize handover control parameters. In \cite{HCPSO}, a self-optimization method was developed for three key handover control parameters, considering factors such as channel conditions to minimize handovers and sustain good throughput. On the other hand, learning-based approaches leverage environmental feedback and machine learning methods to enable data-driven parameter optimization. The work in \cite{HO8} proposed an intelligent dynamic handover parameter optimization strategy to reduce handover failures and redundant handovers. In \cite{GraphHO}, a graph-based approach was used to model handover interactions between overlapping cells. The optimization was carried out using a contextual bandit approach combined with graph convolutional networks. The work in \cite{HO1} proposed an individualistic dynamic handover parameter optimization for 5G networks. It modeled handover control parameters using automatic weight functions tailored to user-specific conditions, such as SINR and cell load. In \cite{HO5}, handover parameters were modeled using machine learning and data mining techniques. The optimization strategy considered radio frequency (RF) conditions at the cell-edge and base station load levels to determine optimal handover parameters. Article \cite{HO7} adopted a Q-learning approach to model and optimize handover parameters between radio and optical networks. The authors in \cite{HO6} applied a temporal-difference learning approach to optimize handover parameters in high-speed railway communications considering the continuously changing environment.

\par However, these works are subject to numerous limitations. First, in an attempt to simplify the handover problem, existing works sacrificed the integrity of the handover mechanism and neglected critical engineering details, resulting in solutions impractical for real-world deployment. Specifically, these works focused solely on either intra-frequency or inter-frequency handovers, failing to address the complexities arising from their coexistence.  Additionally, they did not configure distinct parameters for different handovers objectives, such as coverage-based and priority-based purpose. Second, these studies primarily aimed at maximizing overall system performance while neglecting other crucial factors, such as network load balancing and user fairness. It is very essential to implement load balancing for 5G networks to enhance resource utilization, and support the scalability required by diverse and massive device connectivity. Lastly, the proposed solutions in these works—whether centralized or distributed—relied heavily on high-frequency, wide-area, real-time information exchange, which brings critical issues such as delays and signaling overhead for large-scale networks.

\par Motivated by these limitations, in this paper, we propose to construct an engineering-oriented handover model comprising multiple cells and user equipments (UEs), which fully adheres to the seamless handover process in commercial 5G networks. In addition, to enhance network load balancing through cell cooperation, we introduce a novel decentralized handover parameter optimization scheme based on multi-agent reinforcement learning (MARL), termed MADEHO. The main contributions of our work are summarized as follows:

\begin{itemize}
    \item To effectively optimize the large-scale handover parameters in real 5G networks, we present a comprehensive handover model that fully accounts for engineering details. We begin by categorizing handovers into three types based on different objectives: coverage-based intra-frequency handover, coverage-based inter-frequency handover, and priority-based inter-frequency handover. 
    \item To achieve network load balancing, the problem is formulated as minimizing the standard deviation of cell loads—a joint optimization problem that accounts for the coexistence of different handover types within the same cell and the interdependencies among neighboring cells. Specifically, to facilitate dynamic parameter adjustment in response to user traffic fluctuations, each cell is capable to be aware of the UE distribution and connection status through the Measurement Report (MR) data. To address feedback delays caused by handover latency, we employ a rolling window method for the observation space. 
    \item In the MARL environment, to reduce the cost of real-time global interaction, we implement a distributed training approach based on the consensus approximation of global average load. We demonstrate that, under given assumptions, the approximation error is bounded by a constant. Experimental results show that our scheme significantly outperforms existing methods in terms of load balancing, system throughput, ping-pong rate, and other key metrics.  
\end{itemize}

The remainder of this paper is structured as follows: Section \ref{SystemModel} presents the system model, including scenario descriptions and channel models. Section \ref{marl} outlines the proposed MARL-based optimization scheme. Section \ref{SimulationResults} details the simulation setup and results, while Section \ref{conclusion} concludes the paper.

\section{System Model}\label{SystemModel}

\subsection{Communication Model}
\par Consider an urban macro communication scenario \cite{urbanmacro} consisting of $M$ 5G New Radio (NR) \cite{newradio} cells and $K$ ground UEs, where the cells provide downlinks for UEs. Each cell refers to the effective coverage area provided by a BS on a specific frequency band. The central location of cell $m$ is $q_m=[x_m,y_m]$, whose central frequency and bandwidth are denoted as $F_m$ and $B_m$, separately.

\par We focus on a period which is equally divided into $T$ time slots and denote the location of UE $k$ at time slot $t$ as $l_{k,t}=\left[x_{k,t}, y_{k,t}\right]$. Each UE is served by one cell in a time slot, which can be described with the binary variables $\alpha_{m,k,t}$. $\alpha_{m,k,t}=1$ means the UE $k$ is served by the cell $m$ at time slot $t$ and $\alpha_{m,k,t}=0$ otherwise. So the following constraints hold:
\begin{gather}
    \alpha_{m,k,t} \in \{0,1\},\forall m,k,t,\\
   \sum_{m=1}^M \alpha_{m,k,t} = 1, \forall k,t.
\end{gather}

\par The channels in urban macro communication scenario include both Line-of-Sight (LOS) conditions and Non-Line-of-Sight (NLOS) conditions \cite{LOSNLOS}. According to the standard channel models in \cite{3gpp-tr38901}, the path loss from cell $m$ to UE $k$ at time slot $t$ in $\text{dB}$ can be separately formulated as:
\begin{gather}
PL_{m,k,t}^{LOS}= 28.0 + 22\log_{10}(d_{m,k,t})  + 20\log_{10}(F_m),
     \\
PL_{m,k,t}^{NLOS} = 32.4 + 30\log_{10}(d_{m,k,t})+ 20\log_{10}(F_m),
\end{gather}
where $d_{m,k,t}=\Vert q_m-l_{k,t}\Vert$ represents the distance between the center of cell $m$ and UE $k$ at time slot $t$. We denote the probability of the LOS connection as $Pr^{LOS}_{m,k,t}$ and the probability of the NLOS connection as
 $Pr^{NLOS}_{m,k,t}=1- Pr^{LOS}_{m,k,t}$. The average path loss can be expressed as 
\begin{align}
PL_{m,k,t}^{avg} = PL_{m,k,t}^{LOS} \cdot Pr^{LOS}_{m,k,t}  +PL_{m,k,t}^{NLOS}\cdot Pr^{NLOS}_{m,k,t}.
\end{align}
\par The transmitter antenna gain and receiver antenna gain in $\text{dB}$ are labeled as $g_{tx}$ and $g_{rx}$ respectively. The reference signal received power (RSRP) of UE $k$ from cell $m$ at time slot $t$ is 
\begin{equation}\label{rsrp}
    G_{m,k,t}=h \left( P-PL_{m,k,t}^{avg}+g_{tx}+g_{rx}\right),
\end{equation}
where $P$ is the transmit power of the BSs and $h$ is the fast fading factor following Rayleigh distribution \cite{rayleigh}. 
\par When a cell sends signal to its serving UE, the interference from neighbouring cells of the same frequency significantly reduces the transmission rate. We introduce a binary variable $I_{m,j}$ to record if the two cells $m$ and $j$ use the same frequency. $I_{m,j}=1$ means that they adopt the same frequency and their signal interfere with each other, $I_{m,j}=0$ otherwise. The signal-to-interference-plus-noise ratio (SINR) of the channel linking cell $m$ and UE $k$ can be formulated as
\begin{gather}
  \gamma_{m,k,t}=\frac{G_{m,k,t}}{ \displaystyle\sum_{j=1,j\neq m}^M I_{m,j} G_{j,k,t}+\sigma ^2},
\end{gather}
where $\sigma ^2$ means the power of the additive white Gaussian noise (AWGN) at the receiver. When a cell serves multiple UEs simultaneously, it equally allocates resource blocks (RBs) to them. Without loss of generality, the transmission rate between cell $m$ and UE $k$ at time slot $n$ can be written as
 \begin{align}\label{rate}
 	R_{m,k,t}= \frac{\alpha_{m,k,t}}{\sum_{j=1}^K \alpha_{m,j,t}} B_m \log_2(1+\gamma_{m,k,t}).
 \end{align}
\begin{figure}[!htb] 
\centering 
\includegraphics[width=0.48\textwidth]{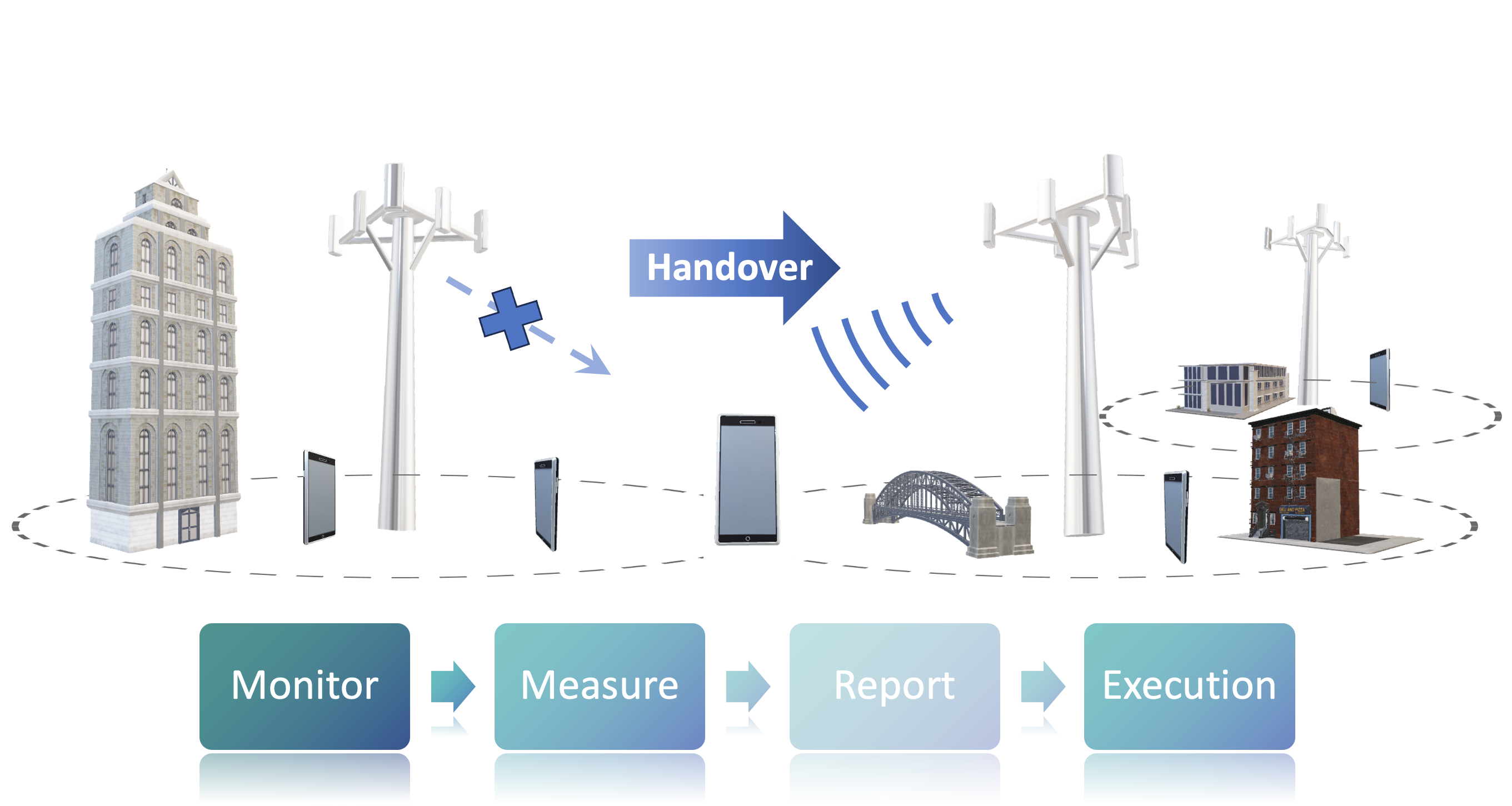} 
\caption{An illustrated 5G cell handover scenario} 
\label{Fig.scenario} 
\end{figure}

\subsection{Handover Procedure} \label{Handover Procedure}

Handover is the process of transferring an active user connection from one cell to another to maintain service continuity as the user moves. As illustrated in Fig. \ref{Fig.scenario}, when a UE moves to the edge of the serving cell, its received signal quality deteriorates, necessitating a switch to a neighboring cell to ensure communication continuity. In 5G protocols, the handover involve four steps \cite{3gpp-ts38331}: monitor, measure, report, and execute. The step details of switching to cells with the same frequency as the serving cell are different from those with different frequencies. For simplicity, we use \textbf{intra-frequency} and \textbf{inter-frequency} to describe whether the frequency of a neighboring cell is the same as that of the serving cell. When a UE $k$ is handed to a new cell $m$ at time slot $t$, we record the updated association by setting $\alpha_{m,k,t}=1$ and $\alpha_{j,k,t}=0, \forall j \neq k$.

The four steps of the Handover Procedure are as follows:

\begin{itemize}
    \item \textbf{Monitor:} When the signal quality of the connected cell remains below a threshold $Z$ for a duration $H_1$, the UE begins to monitor the signal quality of surrounding cells. The signal quality is typically represented by the RSRP $G_{m,k,n}$ as defined in Equation (\ref{rsrp}). This step is crucial for switching to inter-frequency cells, as the UE needs to reconfigure elements such as frequency synthesizers and antennas to monitor signals across multiple frequencies. However, for intra-frequency cells, this step is not required, as the UE can receive multiple intra-frequency signals without adjusting hardware elements.

    \item \textbf{Measure:} Once monitoring is triggered, the UE measures the signal qualities (i.e., RSRPs) of neighboring cells at each time slot. Continuous measurement is necessary because the wireless channel fades and signal strength levels fluctuate over time.
    
    \item \textbf{Report:} If the signal quality of a neighboring cell is satisfactory, the UE reports this event to the network. We use two distinct criteria to assess the performance of intra-frequency and inter-frequency cells. For an intra-frequency cell, the UE reports its measurement results when the neighboring cell's RSRP consistently exceeds that of the serving cell by a margin $U$ over a period $H_2$. In contrast, for an inter-frequency cell, direct RSRP comparison is not rigorous due to fading differences among signals of different frequencies. Instead, the UE reports when the inter-frequency cell's RSRP consistently surpasses a threshold $W$ over the period $H_2$.

    \item \textbf{Execute:} Upon receiving the report, the network initiates the handover process. If multiple neighboring cells are suitable for handover, the network selects the target cell based on the reported measurement information and predefined strategies, typically choosing the cell with the highest signal quality. The network then notifies the original cell, the UE, and the target cell to prepare for the handover. The UE subsequently establishes a connection with the target cell, and the original cell releases the connection. We assume that the handover execution can be completed within one time slot. Therefore, if a UE reports its measurements in time slot $n$, it is transferred to the target cell in time slot $n+1$.
    
\end{itemize}

\par Each cell configures its own handover parameters $Z$, $U$, and $W$. To enhance flexibility and achieve different purposes, we categorize the handovers into three types \cite{3gpp-ts38104}, i.e., \textbf{coverage-based intra-frequency handover (CAH)}, \textbf{coverage-based inter-frequency handover (CEH)}, and \textbf{priority-based inter-frequency handover (PEH)}, with each type adopting independent parameter values. A cell can employ the three handover types simultaneously. We elaborate on the three handover types below and provide a comparison in Table \ref{3type}.

\begin{table*}
\caption{Comparison of Three Cell Handover Types}\label{3type}
\centering
  \begin{tabular}{|m{1.8cm}|m{4.8cm}|m{4.8cm}|m{4.8cm}|}
\hline
\centering\textbf{Handover Type}&
\textbf{coverage-based intra-frequency}
&
\textbf{coverage-based inter-frequency}
&
\textbf{priority-based inter-frequency}
\\\hline
\centering\textbf{Illustration}

&\includegraphics[width=4.8cm,height=3.2cm]{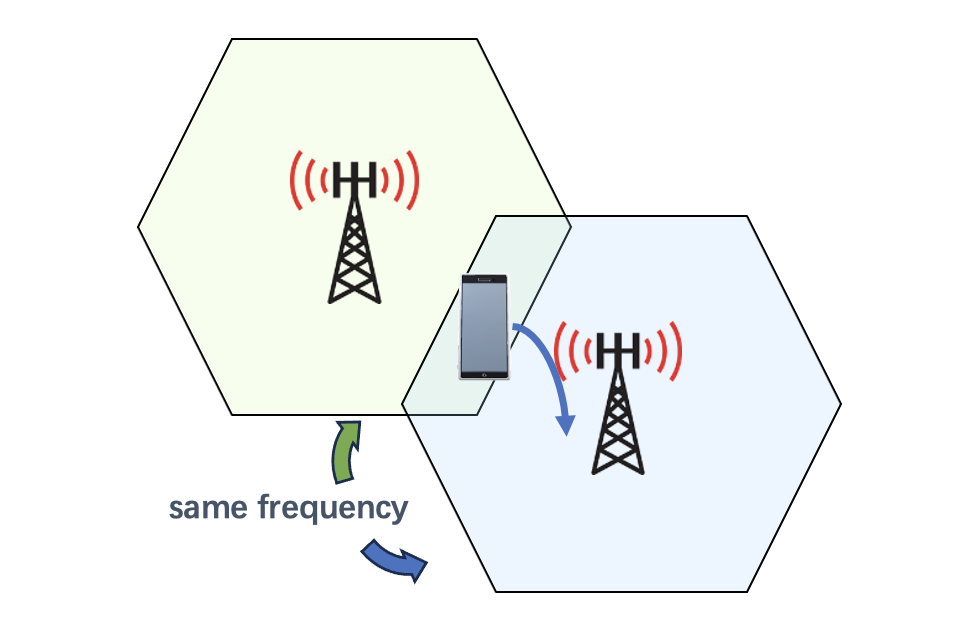}
&\includegraphics[width=4.8cm,height=3.2cm]{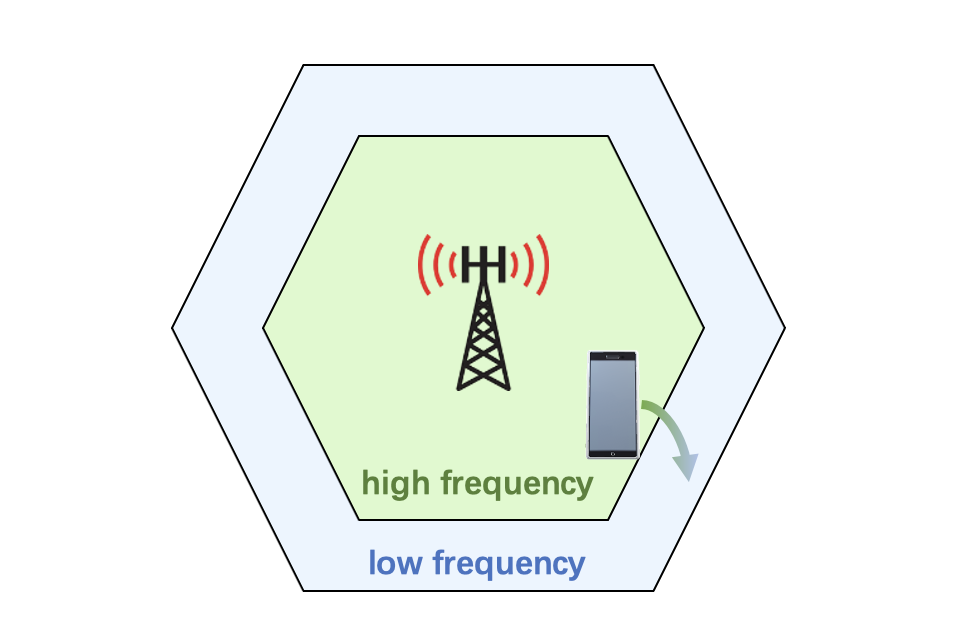}
&\includegraphics[width=4.8cm,height=3.2cm]{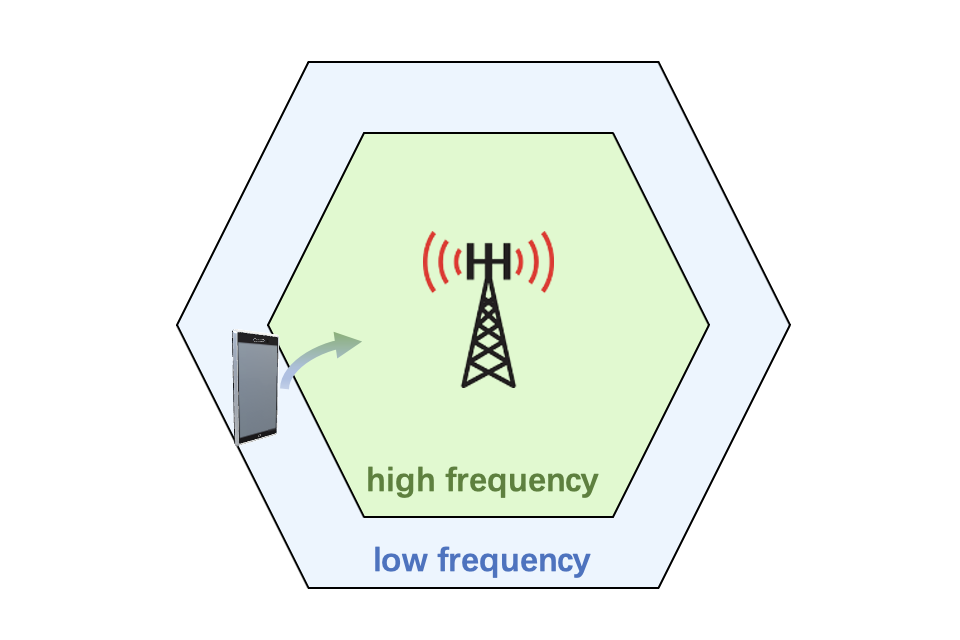}
\\  \hline
 \centering \textbf{Purpose} & maintain connection continuity &  maintain  connection continuity & ensure UE QoS \\
    \hline
  \centering  \textbf{Scenario} & UE located in the junction of two intra-frequency cells & UE moves to lower-frequency cell with  wider coverage & UE switches to higher-frequency cell with larger capacity  \\ 
    \hline
  \centering  \textbf{Direction}  & between intra-frequency cells  & high-frequency to low-frequency  & low-frequency to high-frequency\\
    \hline
  \centering  \textbf{Monitor Condition} & No trigger  & serving cell RSRP lower than $Z_m^{CE}$ for $H_1$ time slots  & serving cell RSRP lower than $Z_m^{PE}$ for $H_1$ time slots\\
    \hline
  \centering  \textbf{Report Condition} & neighboring cell RSRP exceeds that of serving cell by $U_m^{CA}$ for $H_2$ time slots& neighboring cell RSRP exceeds $W_m^{CE}$  for $H_2$ time slots & neighboring cell RSRP exceeds $W_m^{PE}$ for $H_2$ time slots \\
    \hline
  \end{tabular}
  
\end{table*}

\begin{itemize}
    \item \textbf{CAH:} Coverage-based intra-frequency handover. This type of handover is designed to ensure the continuity of connection for UEs between intra-frequency cells, which requires no trigger for the monitor step. The report condition of UE $k$ can be written as:
    \begin{align}
        G_{j,k,\tau }- G_{m,k,\tau }\geq U_m^{CA}, \forall \tau \in \{t,...,t+H_2 \},
    \end{align}
    where cell $m$ is the current serving cell and cell $j$ is an intra-frequency neighbouring cell.
    
    \item \textbf{CEH:} Coverage-based inter-frequency handover. In general, the coverage range of low-frequency cells is larger than that of high-frequency cells. When a UE moves to the edge of a high-frequency cell, it may switch to an adjacent low-frequency cell with a decent signal quality to ensure signal continuity. The trigger condition is
    \begin{align}
        G_{m,k,\tau }\leq Z_m^{CE}, \forall \tau \in \{t,...,t+H_1 \},
    \end{align}
    and the report condition is presented as
    \begin{align}
        G_{j,k,\tau }\geq W_m^{CE}, \forall \tau \in \{t,...,t+H_2 \}, F_j \leq F_m,
    \end{align}
    where cell $m$ is the serving cell of UE $k$, and cell $j$ is a neighbouring cell with lower frequency. 
    \item \textbf{PEH:} Priority-based inter-frequency handover. Due to narrower bandwidth, low-frequency cells have much lower throughput and rates compared to high-frequency cells. To ensure quality of service (QoS), the network migrates UEs from low-frequency cells to high-frequency cells. The trigger and the report conditions are defined as:
    \begin{gather}
        G_{m,k,\tau }\leq Z_m^{PE}, \forall \tau \in \{t,...,t+H_1 \}, \\
        G_{j,k,\tau }\geq W_m^{PE}, \forall \tau \in \{t,...,t+H_2 \}, F_j \geq F_m,
    \end{gather}
     where cell $m$ is the serving cell of UE $k$, and cell $j$ is a neighbouring cell with higher frequency. 
\end{itemize}

\par The 5G protocols \cite{3gpp-ts38331,3gpp-ts38133,3gpp-ts38214,3gpp-ts38215} specify that the cell handover parameters $Z$, $U$, and $W$ must be integers with their unit being \text{dB}. Specifically, $Z$ and $W$ are constrained to range from $-44$ to $-140$, while $U$ is set between $0$ and $96$. Thus we have the following constraints:
\begin{gather}
U_m^{CA} \in \{0,1,...,96 \},  \forall m,\\
    Z_m^{CE},W_m^{CE},Z_m^{PE}, W_m^{PE} \in \{-44,-45,...-140 \}, \forall m.
\end{gather}
Each UE executes the handover type that is reported first. After the handover, all previous measurement records are cleared, and the UE restarts monitoring.

\subsection{Load Balancing Problem Formulation}

\par We use $D_{k,t}$ to represent the size of UE $k$'s request at time slot $t$, and define the load $L_{m,t}$ of cell $m$ at time slot $t$ as the time it takes to finish all the requests of its serving UEs, which can be written as 
 \begin{align}
     L_{m,t}=\frac{1}{\sum_{k=1}^K \alpha_{m,k,t}}\cdot \sum_{k=1}^K\frac{D_{k,t}}{R_{m,k,t}}\cdot d_t,
 \end{align}
so the standard deviation $\Gamma_t$ of all cells' loads at time slot $t$ is calculated by the following equation 
\begin{align}
\Gamma_t = \sqrt{\frac{1}{M} \sum_{m=1}^{M} \left( L_{m,t} - \bar{L}_t \right)^2},
\end{align}
where $\bar{L}_t=\frac{1}{M} \sum_{i=1}^{M} L_{i,t}$ is the average load of all cells at time slot $t$. Denote handover parameters set as $\mathcal{U}=\{U_m^{CA}|\forall m \}$, $\mathcal{Z}=\{Z_m^{CE},Z_m^{PE}|\forall m \}$, and $\mathcal{W}=\{W_m^{CE},W_m^{PE}|\forall m \}$. Our objective is to achieve load balancing in 5G cellular networks by minimizing the cumulative standard deviation of cell loads across the entire operation period, with respect to the parameters in $\mathcal{U}$, $\mathcal{Z}$, and $\mathcal{W}$. This optimization problem can be formulated as: 
 \begin{subequations}\label{p1}\allowdisplaybreaks
  \begin{align}
   & \min_{\mathcal{U},\mathcal{Z},\mathcal{W}} \sum_{t=1}^T \Gamma_t, \\
   \mbox{s.t.  }
 &  \alpha_{m,k,t} \in \{0,1\},\forall m,k,t,\\
  & \sum_{m=1}^M \alpha_{m,k,t} = 1, \forall k,t,\\
  &U_m^{CA} \in \{0,1,...,96 \},  \forall m,\\
   & Z_m^{CE},W_m^{CE},Z_m^{PE}, W_m^{PE} \in \{-44,-45,...-140 \}, \forall m.
\end{align}
\end{subequations}

 \par This is an NP-hard combinatorial optimization problem with highly coupled variables and non-convex formulations. Due to the high density of cell deployment and their mutual influence in real scenarios, a large number of cells and parameters need to be considered jointly. The sheer scale and complexity of the variable space make it infeasible to solve the problem using precise methods within a limited time. Therefore, we propose a multi-agent-reinforcement-learning (MARL)-based decentralized handover parameter optimization scheme (MADEHO) for load balancing in Section \ref{marl}.  

\begin{figure*}[htbp] 
\centering 
\includegraphics[width=\textwidth]{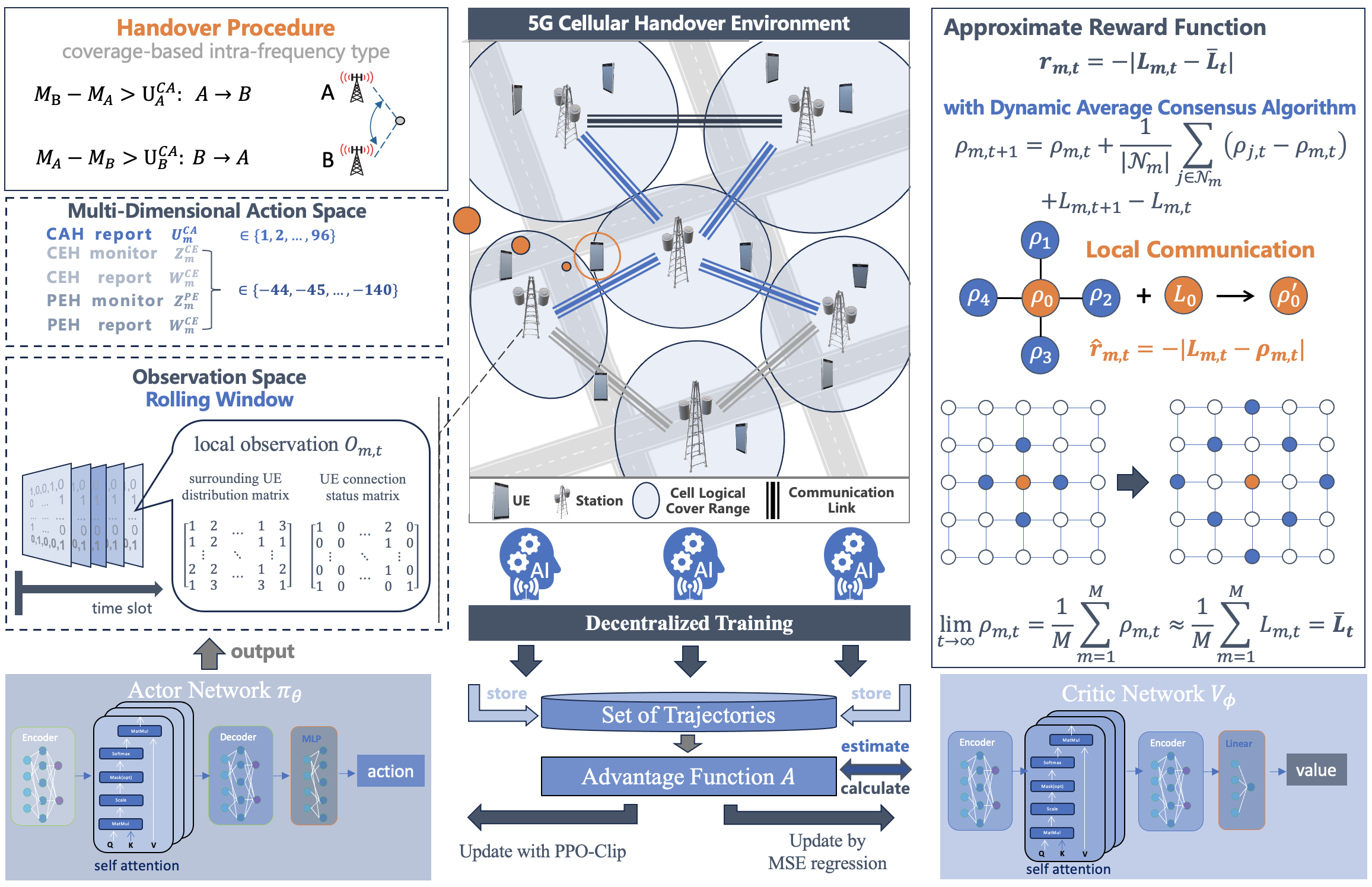}
\caption{Cells observe UE distribution and connection status to adapt handover parameters. They achieve decentralized training via local communication, using PPO to update policy and value networks for handling large action spaces.} 
\label{Fig.5GRL} 
\end{figure*}
\section{MARL-Based Decentralized Handover Parameter Optimization for Load Balancing}\label{marl}
\par In the 5G cellular network, we model each cell as an agent, and all agents collaboratively learn the optimal handover parameters using a Proximal Policy Optimization (PPO)-based MARL approach. Specifically, we first define the reinforcement learning environment. Subsequently, we design a decentralized training framework that relies on local communication to approximate global rewards using a dynamic average consensus algorithm. Finally, we introduce the policy optimization procedure to update the agents' policies. The main details of MARL environment definition and the process of policy optimization are demonstrated in Fig. \ref{Fig.5GRL}.

\subsection{Environment Definition}
\par The cell handover parameters determine the conditions under which UEs are transferred between cells, thereby affecting the logical boundaries and coverage areas of cells. A cell's logical boundary is influenced by both its own handover parameters and those of its neighboring cells. Consequently, cells ought to observe the load distribution and collaboratively adjust their handover parameters to achieve load balancing. This process can be viewed as a decentralized Partially Observable Markov Decision Process (Dec-POMDP) \cite{pomdp}. We define the Dec-POMDP by the tuple $ (S, \mathcal{A}_{\times}, \mathcal{P}, \mathcal{R}_{\times}, \Omega_\times, \gamma)$:

\begin{itemize}
    \item \textit{S} is the state space, which includes the distribution of all UEs and their association relationship with cells.
    \item $\mathcal{A}_{\times} $ is the joint action space, where each agent selects an action from its local action space. In this environment, each agent uses a multi-dimensional action space to select values for five discrete handover parameters $U^{CA}_m, Z^{CE}_m, Z^{PE}_m, W^{CE}_m, W^{PE}_m$ mentioned in Section \ref{Handover Procedure}. The multi-dimensional action space for agent $m$ at time slot $t$ can be written as 
    \begin{equation}
        a_{m,t} = \{0, 1, 2, \dots, 96\} \times \{-44, -45, \dots, -140\}^4
    \end{equation}
    
    \item $\mathcal{P}$ is the transition probability function (or transition
matrix), which is generated by the random state of the channels and handover processes.
    \item $\mathcal{R}_{\times}$ is the local reward function designed as negative of the absolute difference between the cell's current load and the average load $\bar{L}_t$ of all cells, which can be written as 
    \begin{equation} \label{originalreward}
        r_{m,t}=-\lvert L_{m,t}-\bar{L}_t \rvert.
    \end{equation}

\item $\Omega_\times$ is the joint observation space. Each agent observes the surrounding UE distribution together with their connection status. We divide the map into several grids, each grid having a side length of $\nu$. Each cell can observe a square area centered on its own grid, with each side containing $ 2\kappa + 1 $ grids. The surrounding UE distribution is recorded in a $ (2\kappa + 1 ) \times (2\kappa + 1 ) $ matrix $\mbox{UED}_{m,t}$, which describes the number of UEs in each grid of the map. The connection status of UEs are written in a $ (2\kappa + 1 ) \times (2\kappa + 1 ) $ matrix $\mbox{CSM}_{m,t}$ to show how many UEs in each grid are connected to the cell. Considering the fluctuations caused by signal degradation and the time it takes to switch, each agent is allowed to observe the surrounding UE distribution, and their connection status over the preceding $\eta$ time slots with a rolling window method. The observation of agent $m$ at time slot $t$ is:
\begin{equation}
\begin{split}
     o_{m,t}=\{\mbox{UED}_{t-\eta+1},\mbox{CSM}_{t-\eta+1},...,\mbox{UED}_t, \mbox{CSM}_t\}.
\end{split}
\end{equation}

\item $\gamma \in [0,1]$ is the discount factor, which determines the
importance of future rewards. Then the actual return \( g_{m,t} \) after action $a_{m,t}$ can be formulated as
\begin{equation}
    g_{m,t}  = \sum_{\tau=t}^{T} \gamma^{\tau-t} r_{m,\tau},
\end{equation}
meaning the accumulated reward of $m$ after action $a_{m,t}$.
\end{itemize}

\subsection{Reward Approximation Based on Local Communication}

 The local reward function (\ref{originalreward}) requires real-time global information to compute the average load of all cells, which is impractical in real-world scenarios due to the high communication overhead and complexity. This challenge is further compounded by the fact that global information exchange is a significant bottleneck in MARL, as it hampers the scalability and efficiency of decentralized training. To address this issue, as depicted in Fig. \ref{Fig.5GRL}, we employ local communication rather than global communication to facilitate decentralized training. Although each cell cannot directly calculate the average load $\bar{L}_t$ of all cells when using local communication, we leverage the dynamic average consensus algorithm \cite{dynamicaverageconsensus} to approximate it effectively. This approach not only reduces communication costs but also enhances the robustness and adaptability of the system in dynamic environments.

 \par We model the cellular network as a graph $G$, where each cell is a vertex. When the distance of two cells does not exceed $\chi$, they are considered as neighbors and within each other's local communication range. Each cell maintains an estimate of the average load $\rho$ and, at each time slot, it exchanges its estimate of the average load with its neighbors and subsequently update it using the following formulation:

    \begin{equation}\label{update}
         \rho_{m,t+1}=\rho_{m,t}+L_{m,t+1}-L_{m,t}+\sum_{j \in \textrm{N}_m}\frac{1}{\lvert\textrm{N}_m \lvert}(\rho_{j,t}-\rho_{m,t}),
    \end{equation}
    where $\rho_{m,t}$ is cell $m$'s estimation of the average load at time slot $t$, $\textrm{N}_m$ is the set containing neighbors of cell $m$, $\lvert\textrm{N}_m \lvert$ is the number of its neighbors, and $\rho_{m,1}=L_{m,1}$. Then we can approximate the local reward function as
    \begin{equation} \label{approreward}
        r_{m,t} \approx -\lvert L_{m,t}-\rho_{m,t} \rvert.
    \end{equation}

To ensure the feasibility of approximating the average load, we analyze the error of the dynamic average consensus algorithm and prove that under the assumption of a limited load change rate, its error has a constant upper bound. 

\begin{assumption} \label{assumption}
It is assumed that $G$ is a connected and irreducible graph. For cell $m$ at each time slot $t$, the absolute values of the real time load $L_{m,t}$ and its change rate $\dot{L}_{m,t}=L_{m,t}-L_{m,t-1}$ are both bounded. Specifically, there exists a constant $\zeta > 0$ such that
\begin{gather}
    \sup|L_{m,t}| \leq \zeta < \infty, \forall m, t,\\
\sup|\dot{L}_{m,t}| \leq \zeta < \infty, \forall m, t.
\end{gather}
\end{assumption}

\begin{theorem}
Under Assumption \ref{assumption}, the dynamic average consensus algorithm ensures that the error between the estimate $\rho_{m,t}$ and the average load $\bar{L}_t$ is bounded by a constant, which can be specified as:
\begin{equation}
  |\rho_{m,t} - \bar{L}_t| \leq \frac{3-\lambda}{1-\lambda}\zeta, \forall m, t,
\end{equation}
where $\lambda$ is the spectral radius of the consensus matrix $\omega$ associated with the graph $G$. In addition, when $t \to \infty$, the error bound is lower, which is:
\begin{equation}
\lim_{t \to \infty} |\rho_{m,t} - \bar{L}_t| \leq \frac{2\zeta}{1-\lambda}, \forall m.
\end{equation}
\end{theorem}

\begin{proof}
Please see the Appendix.
\end{proof}

\subsection{Decentralized Policy Optimization}

The action space of each agent is large due to the wide ranges of the handover parameters. To avoid directly calculating the Q-value for each action, we utilize PPO \cite{ppo} as the policy optimization algorithm, which directly optimizes the policy distribution. In this environment, each agent $m$ trains its own policy network $\pi$ parameterized by $\theta_m$ and value network $V$  parameterized by $\phi_m$ with PPO. The training process is decentralized, since each agent only depends on its local observation and local reward. The inputs to both the policy network $\pi_{\theta_m}$ and the value network $V_{\phi_m}$ are the agent's observation $o_{m,t}$. The policy network $\pi_{\theta_m}$ selects actions to maximize cumulative rewards, while the value network $V_{\phi_m}$ estimates the expected value of a certain state or observation. 
\par  We denote the parameters of $m$ after $i$ episodes of update as $\theta^i_m$ and $\phi^i_m$. PPO updates the policy network by maximizing a clipped advantage function. The advantage function $A_{m,i,t}$ measures the additional benefit for agent $m$ of taking a specific action $a_{m,t}$ in a particular state $o_{m,t}$ compared to the average situation estimated by $V_{\phi^i_m}$, which is formulated as: 
\begin{equation}
    A_{m,i,t} = \sum_{l=t}^{T} (\gamma \xi)^{l-t} \delta_{m,i,t}
\end{equation}
 where $\delta_{m,i,t} = r_{m,t} + \gamma V_{\phi^i_m} (o_{m,t+1}) - V_{\phi^i_m} (o_{m,t})$ is the temporal difference (TD) error at time slot $t$ and $\xi$ is a GAE parameter that controls the trade-off between bias and variance. So the optimization objective $L^{policy}_{m,i,t}$ is 
 \begin{equation}
   L^{policy}_{m,i,t} =\min \left( \frac{\pi_{\theta_m}(a_{m,t}|o_{m,t})}{\pi_{\theta^i_m}(a_{m,t}|o_{m,t})}  A_{m,i,t}, g(\epsilon, A_{m,i,t}) \right),
\end{equation}
 where
\begin{equation}
    g(\epsilon, A) =
\begin{cases}
(1 + \epsilon) A & A \geq 0 \\
(1 - \epsilon) A & A < 0
\end{cases}
\end{equation}
and \( \epsilon \) is a hyper-parameter that controls the range of allowed policy updates. The advantage function is clipped in $L^{policy}_{m,i,t}$ to prevent large deviations from the current policy. Since the policy network aims to maximize the clipped advantage function, the policy parameters $\theta^i_m $ are updated by taking multiple steps of mini-batch Gradient Ascent \cite{sgd} as follows:
\begin{equation}
    \theta^{i+1}_m = \theta^i_m + \alpha_\theta \nabla_{\theta_m} L^{policy}_{m,i,t}.
\end{equation}

\par The value network is updated by minimizing the mean squared error between the predicted value $ V_{\phi^i_m}(o_{m,t}) $ and the actual return $g_{m,t}$. The loss function for the value network is 
\begin{equation}
   L^{value}_{m,i,t} = \frac{1}{2} \mathbb{E} \left[ \left( V_{\phi^i_m}(o_{m,t}) - g_{m,t} \right)^2 \right],
\end{equation}
while the value parameters $\phi_m$ are updated by minimizing the loss function $L^{value}_{m,i,t}$ using gradient descent:
\begin{equation}
    \phi^{i+1}_m = \phi^{i}_m - \alpha_\phi \nabla_{\phi_m} L^{value}_{m,i,t}.
\end{equation}
In order to elaborate on the MADEHO algorithm, we present the details of the training process in Algorithm \ref{algorithm1}. 

\begin{algorithm}
\caption{MARL-based Decentralized Handover Parameter Optimization (MADEHO)}\label{algorithm1}
\begin{algorithmic}[1]
\STATE \textbf{INPUT}: Initial parameters $\theta^0$, $\phi^0$, learning rates $\alpha_\theta$, $\alpha_\phi$
\STATE \textbf{OUTPUT}: Optimized policy parameter $\theta^*_m$ continuously
\STATE $\theta^0_m=\theta^0, \phi^0_m=\phi^0, \forall m$ 
\FOR{each episode}
\STATE Set cell locations and cell frequencies
\STATE Generate UE trajectories
\STATE Create a new empty set $\mathbf{TS}_m$ for each agent
\FOR{$t=1$ to $T$}
\FOR{each agent $m$}
\IF{unterminated}
    \STATE Get observation $o_{m,t}$ 
    \STATE Select action $a_{m,t}$ using $\pi_{\theta^i_m}$
    \IF{ $t==T$}
    \STATE terminate
    \ENDIF
    \STATE Get reward $r_{m,t}$ and new observation $o_{m,t+1}$ 
    \STATE Store $\left( o_{m,t}, a_{m,t}, r_{m,t}, o_{m,t+1} \right)$ in $\mathbf{TS}_m$
    \ENDIF
    \ENDFOR
    \ENDFOR
    \FOR{each agent $m$}
        \STATE Update parameters of policy network and value network using trajectories in $\mathbf{TS}_m$ :
        \[
        \theta^*_m \leftarrow \theta^i_m + \alpha_\theta \nabla_{\theta_m}  L^{policy}_{m,i,\mathbf{TS}_m}
        \]
        \[
        \phi^*_m \leftarrow \phi^i_m - \alpha_\phi \nabla_{\phi_m} L^{value}_{m,i,\mathbf{TS}_m}
        \]
    \ENDFOR
    \ENDFOR
\end{algorithmic}
\end{algorithm}

\par To enhance the ability to perceive interference and evaluate link capacity, we introduce a self-attention mechanism \cite{attention} into the policy network and the value network. The attention layer takes in encoded observation embeddings $E$ and trains ego query $W_Q$, collective key $W_K$ and collective value $W_V$. We denote $d_k$ as the embedding dimension and $h(\cdot)$ as the activate function. Thus the weighted attention is obtained by:
\begin{equation}
\text{Attention} = \text{softmax}\left(\frac{(W_Q E) (W_K E)^T}{\sqrt{d_k}}\right)h(W_V E).
\end{equation}

\par The concurrent nature of three types of handover events and frequent interactions among multiple cells and multiple UEs result in an extremely complex multi-agent environment. To clearly illustrate the implementation of distributed optimization in this environment, we present the optimization process in the form of a flowchart as shown in Fig. \ref{Fig.flowchart}. At the beginning of each time slot, each cell observes its surrounding environment and uses its own actor network to set the handover parameters. Next, all UEs follow the handover process. Specifically, if the monitor step of a UE is triggered, it continuously measures the RSRP of nearby cells and check whether the \textbf{CAH}, \textbf{CEH}, and \textbf{PEH} conditions are met. If any condition is satisfied, a handover is executed. At the end of each time slot, each cell stores its observation data and calculates the local reward with Equation (\ref{approreward}). Once an episode is terminated, the cells update their own Actor and Critic parameters using the MADEHO algorithm. When the training reaches the predefined number of episodes, the Actor networks of all cells are output.

\begin{figure}[!htb] 
\centering 
\includegraphics[width=0.5\textwidth]{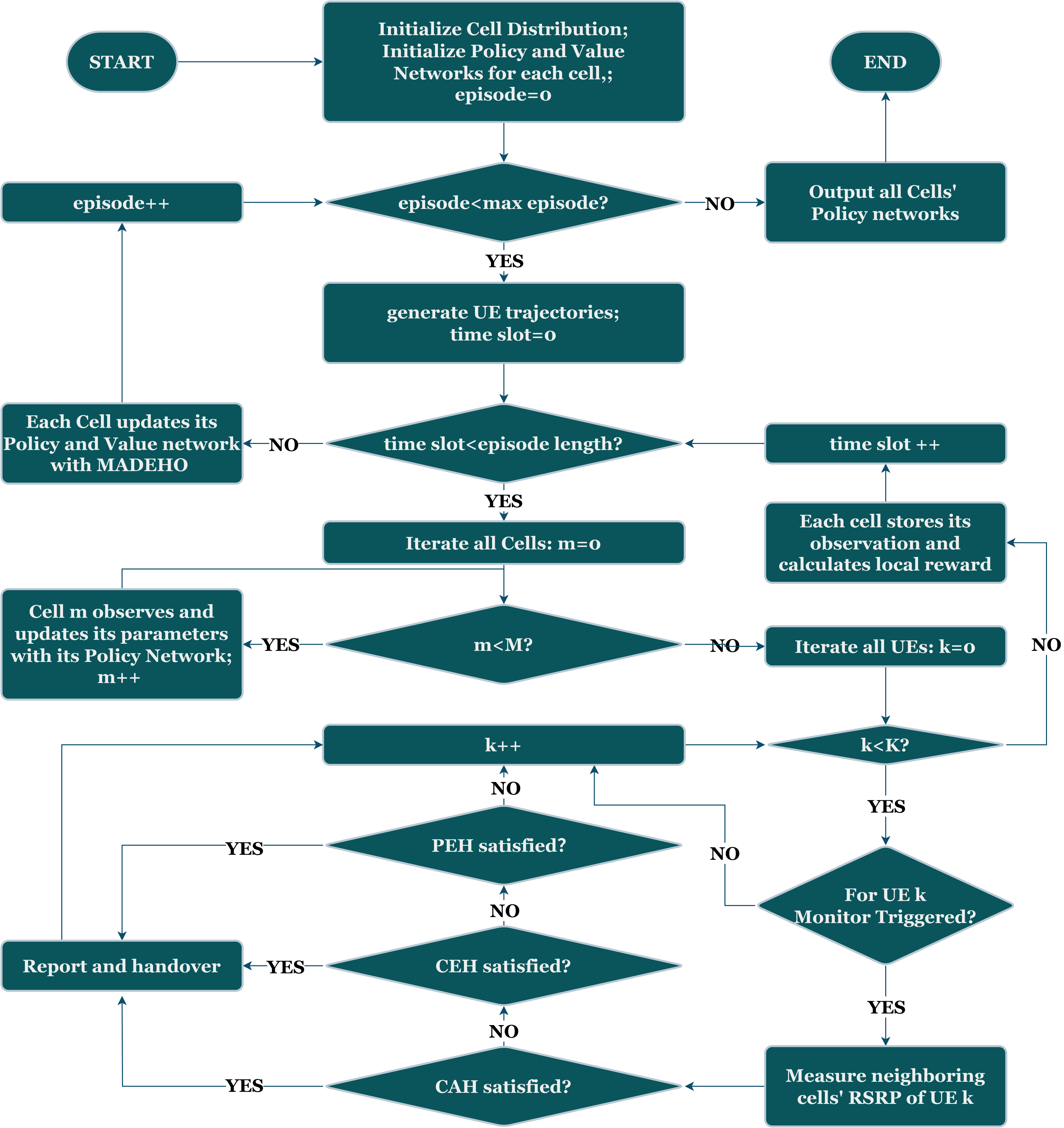} 
\caption{The complete process of agent-environment interaction and policy optimization in the form of a flowchart.  } 
\label{Fig.flowchart} 
\end{figure}

\section{Simulation Results} \label{SimulationResults}
In this section, we provide simulation results to demonstrate the effectiveness of our proposed scheme. 
We begin by introducing the simulation setup, and we investigate the factors that influence the difficulty of load balancing. We then compare our scheme with three other schemes on load balancing and other performance metrics. Finally, we analyze the impact of BS frequency assignment on load balancing.

\begin{figure}[htbp] 
\centering 
\includegraphics[width=0.48\textwidth]{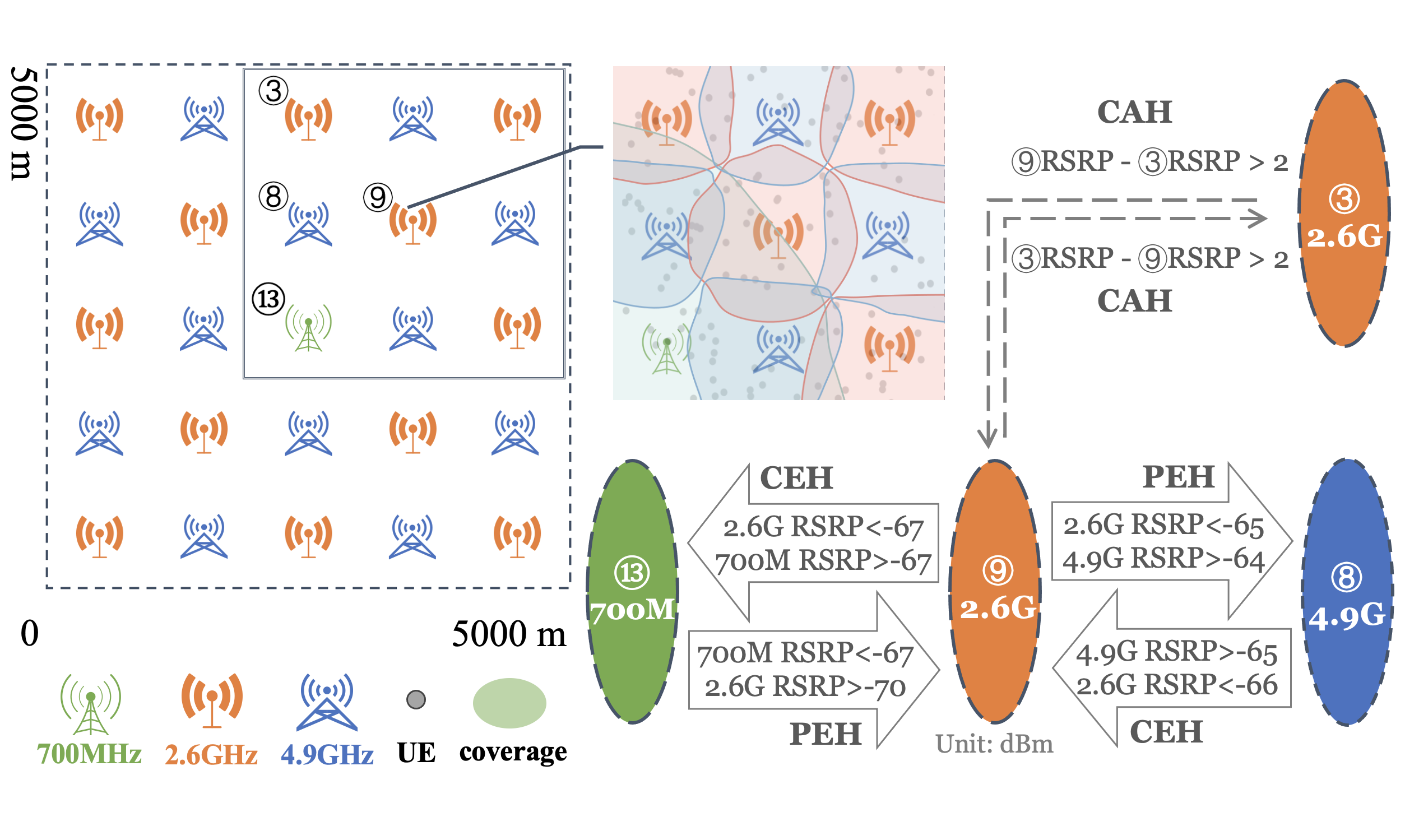} 
\caption{The representative case shows the simulation setup and part of the optimized handover parameters with our proposed scheme.} 
\label{Fig.case} 
\end{figure}
\subsection{Simulation Setup}
We consider a $\SI{5000}{\meter} \times \SI{5000}{\meter}$ area divided into a $5 \times 5$ grid of squares where $500$ moving UEs require BS connection. As shown in Fig. \ref{Fig.case}, each square has a BS deployed at its center, providing a cell with a specific frequency band. According to  \cite{3gpp-ts38104}, in 5G Frequency Range 1 (FR1), the three main frequency bands in use are 700 MHz, 2.6 GHz, and 4.9 GHz, with bandwidths of 10 MHz, 40 MHz, and 50 MHz, respectively. The frequency band of each cell is illustrated in Fig. \ref{Fig.case}.

In real-world situations, user movement is typically purposeful while also influenced by random factors. To simulate the characteristic, we use Ornstein-Uhlenbeck (OU) process \cite{OU} to generate the trajectories of UEs, which is written as:
\begin{align}
    l_{k,t+1} = l_{k,t} + \iota (\mu - l_{k,t} )d_t + \sigma \sqrt{d_t} \zeta,
\end{align}
where $\iota, \mu, \sigma, \zeta $  are all UE-related parameters representing mean speed, mean determined position, volatility, and random noise that follows a standard normal distribution. At the beginning of each period, we generate each UE's initial position $l_{k,1}$ and mean determined position $\mu_k$ following a normal distribution $(\mu_x, {\sigma_x}^2)$ and $(\mu_y, {\sigma_y}^2)$ in the $x$- and $y$-directions, respectively. $ \mu_x, \mu_y \in [500, 4500] $ and $ \sigma_x, \sigma_y \in [0.5, 5] $. The mean reversion speed $ \iota $ is set as $ \iota = \frac{1}{T} $ to ensure UEs reach their destination at the end of the period. The volatility value $ \zeta $ is set to 0.1. For each period, we adjust the values of $ \mu_x, \sigma_x, \mu_y, \sigma_y $ to simulate different UE trajectories.
\par In the training process, our proposed scheme is updated by $10000$ episodes, each containing $100$ time slots. Then we test our scheme over another $10,000$ episodes. The detailed parameter settings are provided in Table \ref{table1}. Fig. \ref{Fig.case} shows a representative case where the handover parameters are optimized by our scheme. In this case, $\mu_x$ and $\mu_y$ are set as $2500$, while $\sigma_x$ and $\sigma_y$ are assigned to be $4$. Take Cell $9$ for instance, the real-time surrounding UE distribution and its coverage are depicted. The coverages of different cells may overlap since the handover parameters only determine when a UE leaves and it may shortly switch back. The coverage of Cell $13$ is larger than those of the other cells since low-frequency cell signals attenuate more slowly. The handover parameters of Cell $9$ for CAH, CEH, and PEH are separately listed. Its PEH parameters are higher than CEH parameters, which reveals that migrating to a higher-frequency cell requires more caution due to the disturbance caused by rapid signal attenuation.

\begin{table}[h]
\centering
\caption{Simulation Parameters}\label{table1}
\begin{tabular}{|l|l|}
\hline
\textbf{Parameter}                                   & \textbf{Value}                                      \\ \hline
number of cells & 25
\\ \hline
number of UEs & 500         
\\ \hline
available frequency (\SI{}{\giga\hertz}) & 0.7, 2.6, 4.9
\\ \hline
corresponding bandwidth (\SI{}{\mega\hertz}) & 10, 40, 50
\\ \hline
length of a time slot $d_t$  & \SI{200}{\milli\second}  
\\ \hline
number of time slot $N$ & 100
\\ \hline
BS transmit power $P$ &  \SI{45}{\text{dBm}}
\\ \hline
fast fading factor $h$ & $\sim \text{Rayleigh}(\sigma = 2)$
\\ \hline
sender antenna gain $g_{tx}$ & \SI{10}{\text{dB}}
\\ \hline
receiver antenna gain $g_{rx}$ & \SI{1}{\text{dB}}
\\ \hline
time to trigger monitor $H_1$ & \SI{5}{\text{ time slot}}
\\ \hline
time to trigger report $H_2$ & \SI{3}{\text{ time slot}}
\\ \hline
AWGN $\sigma^2$ & \SI{-110}{\text{dBm}}
\\ \hline
UE request $D_{k,n}$ & \SI{1}{\mega\byte}
\\ \hline
discount factor $\gamma$ & 0.99
\\ \hline
neighbour distance threshold $\chi$ & \SI{2000}{\meter}
\\ \hline
rolling window length $\eta$ & \SI{5}{\text{ time slot}}
\\ \hline
observation grid length $\nu$ & \SI{200}{\meter}
\\ \hline
grid number per side $2\kappa+1$ & 15
\\ \hline
learning rate $\alpha_\theta, \alpha_\phi$ & 0.001, 0.001
\\ \hline
GAE parameter $\xi$ &0.95
\\ \hline
GAE hyper-parameter $\epsilon$ & 0.1
\\ \hline
mean reversion speed $ \iota $ & 0.01
\\ \hline
volatility value $ \zeta $ & 0.1
\\ \hline
mean of UE distribution $\mu_x, \mu_y$ &  random in [500, 4500]
\\ \hline
std of UE distribution $ \sigma_x, \sigma_y$ & random in [0.5, 5]
\\ \hline
number of training episodes & 10000
\\ \hline
number of testing episodes & 10000
\\ \hline
\end{tabular}
\end{table}
\subsection{Load Balancing Performance Comparison} \label{CompareScheme} 
To accurately assess and control the difficulty of load balancing, in this part we first introduce two influencing factors: UE distribution balancing degree and UE mobility speed. Next, we evaluate the performance of our scheme and baselines in achieving load balancing.
\subsubsection{UE distribution balancing degree}
A cell's load is related to the surrounding UE density since UEs can only connect to nearby cells, so the UE distribution directly affects load balancing difficulty. We adopt the average standard deviation of the number of UEs in each grid across time slots to measure the UE distribution balancing degree. The number of grids is $\psi=\frac{5000}{\nu}$. We denote ${CT}_{i,t}$ as the number of UEs in grid $i$ at time slot $t$. For simplicity, we abbreviate the standard deviation as `std'. So the UE distribution std is:
\begin{equation}
\begin{split}
    & \textbf{UE distribution std} \\
    & = \frac{1}{T}\sum_{t=1}^T \sqrt{\frac{1}{\psi} \sum_{i=1}^{\psi} \left( {CT}_{i,t} - \frac{1}{\psi} \sum_{j=1}^{\psi} {CT}_{j,t} \right)^2}.
\end{split}
\end{equation}

\subsubsection{UE mobility speed} \label{mobilityspeed}
Since the strength of the received signal is positively correlated with the distance between BSs and UEs, cell handover can be frequently triggered when the UEs move fast, allowing cells to balance their loads. On the other hand, when the location of a UE changes too quickly, it may lose the original signals and fail to conduct a successful handover. In this part, we define \textbf{Average UE speed} as the mean of all UEs' instantaneous speeds, which is expressed as:
\begin{equation}
\begin{split}
   & \textbf{average UE speed}\\
   & =\frac{1}{M}\sum_{k=1}^K\frac{1}{N}\sum_{n=1}^N \left(\frac{||l_{k,n+1}-l_{k,n} ||}{d_t}\right).
\end{split}
\end{equation}
\par Fig. \ref{Fig.3D} is a 3D plot that reflects the impact of UE mobility speed and UE distribution balancing degree on network load balancing using our scheme. Given an average UE speed, the cell load std rises with the growth of UE distribution std, which indicates that the more balanced the UE distribution is, the easier it is to achieve load balancing. Given the UE distribution std, the cell load std first decreases and then increases with the average UE speed going up. This is because moderately increasing UE mobility can raise the frequency of handover, thereby promoting load balance; however, when the UEs move too fast, the signal quality becomes unstable, which hinders handovers and load balancing.

\begin{figure}[htbp] 
\centering 
\includegraphics[width=0.45\textwidth]{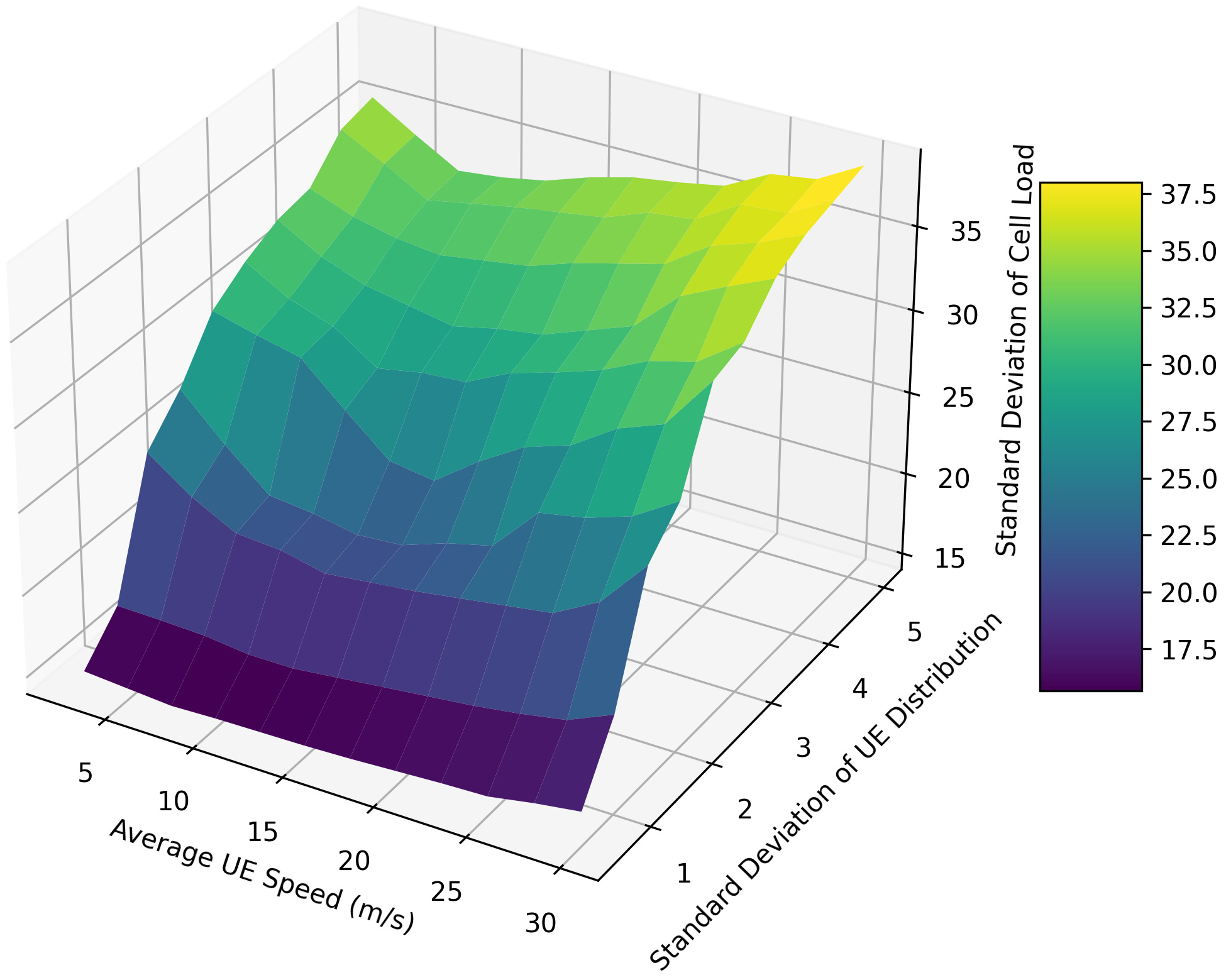}
\caption{Cell load std vs. average UE speed and UE distribution std.} 
\label{Fig.3D}
\end{figure}
\vspace{-10pt}

In the simulation experiments, we compare the following four handover parameter optimization schemes:
\begin{itemize}
\item \textbf{MADEHO (Ours):} an MARL-Based scheme for decentralized handover parameter optimization. 
    
\item \textbf{GraphHO \cite{GraphHO}:} A graph-based scheme to optimize handover parameters, which leverages a graph-based model with primal-dual GCNs and contextual bandits.

\item \textbf{DDSO \cite{DDSO}:}  A data-driven scheme for optimizing handover parameters, which combines SHAP-based sensitivity analysis for efficient data sampling with machine learning models to capture parameter-KPI relationships.

\item \textbf{HCPSO \cite{HCPSO}:} A handover control parameter self-optim-ization scheme using real-time data on user mobility, channel, and system parameters by applying analytical models based on these measurements.
\end{itemize}

\par Fig. \ref{Fig.UEDistribution} shows how the cell load std changes with the increase of the UE distribution std. The average UE speed is limited in $[24, 26] \,\text{m/s}$. The lower the cell load std, the more balanced the network. Our scheme decreases the cell load std by $4.16\%$, $14.43\%$, and $15.7\%$ compared with others, respectively.

\begin{figure}[htbp] 
\centering 
\includegraphics[width=0.45\textwidth]{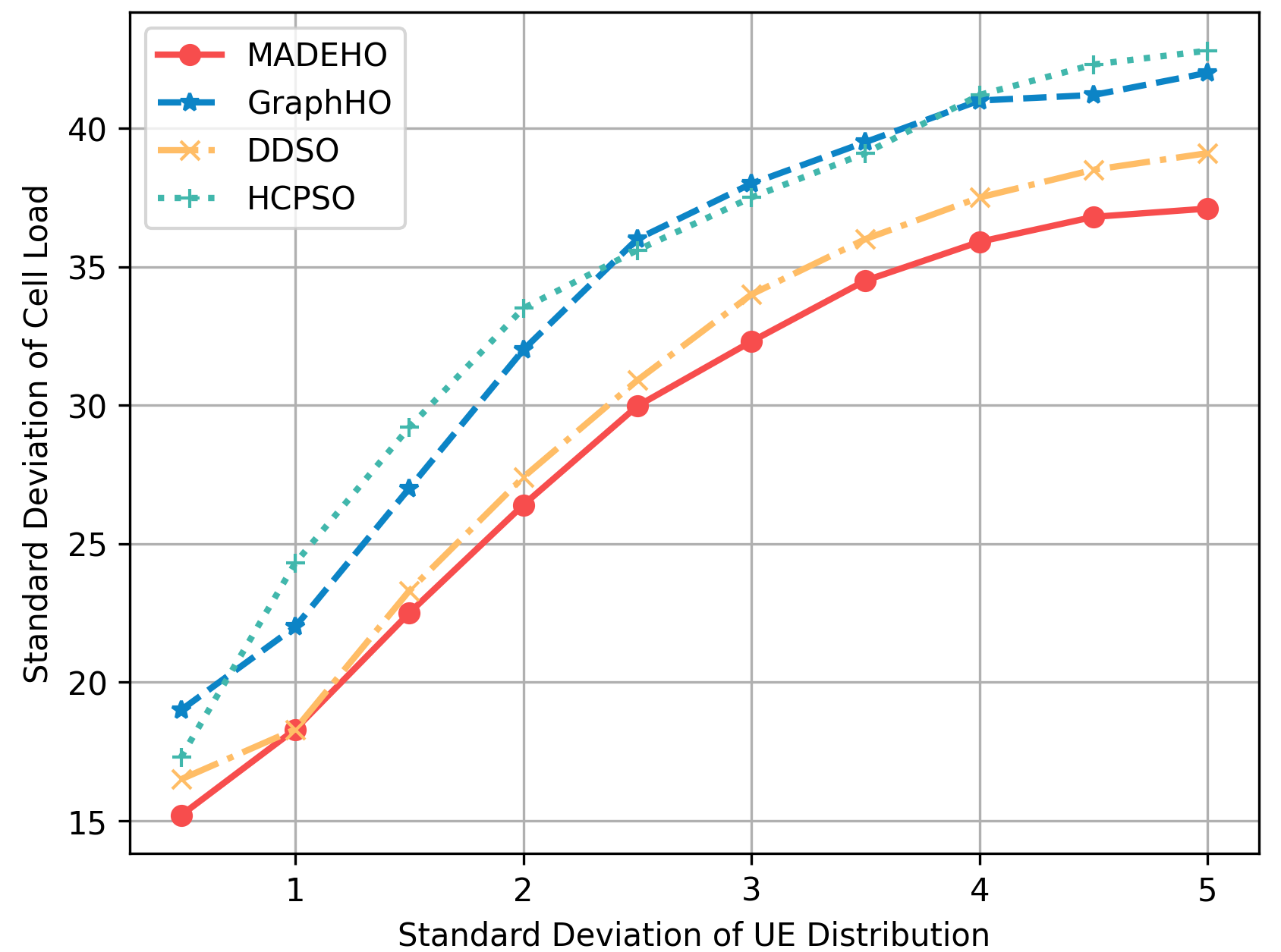}
\caption{UE distribution std vs. cell load std with the average UE speed limited in $[24, 26] \,\text{m/s}$.  }
\label{Fig.UEDistribution}
\end{figure}

\par Fig. \ref{Fig.MovingSpeed} represents the change of the cell load std with the rise of the average UE speed. The UE distribution std is controlled within $[2.9, 3.1]$. The cell load std is separately lower than those of the other schemes by $3.95\%$, $14.14\%$, and $17.0\%$.

\begin{figure}[htbp] 
\centering 
\includegraphics[width=0.45\textwidth]{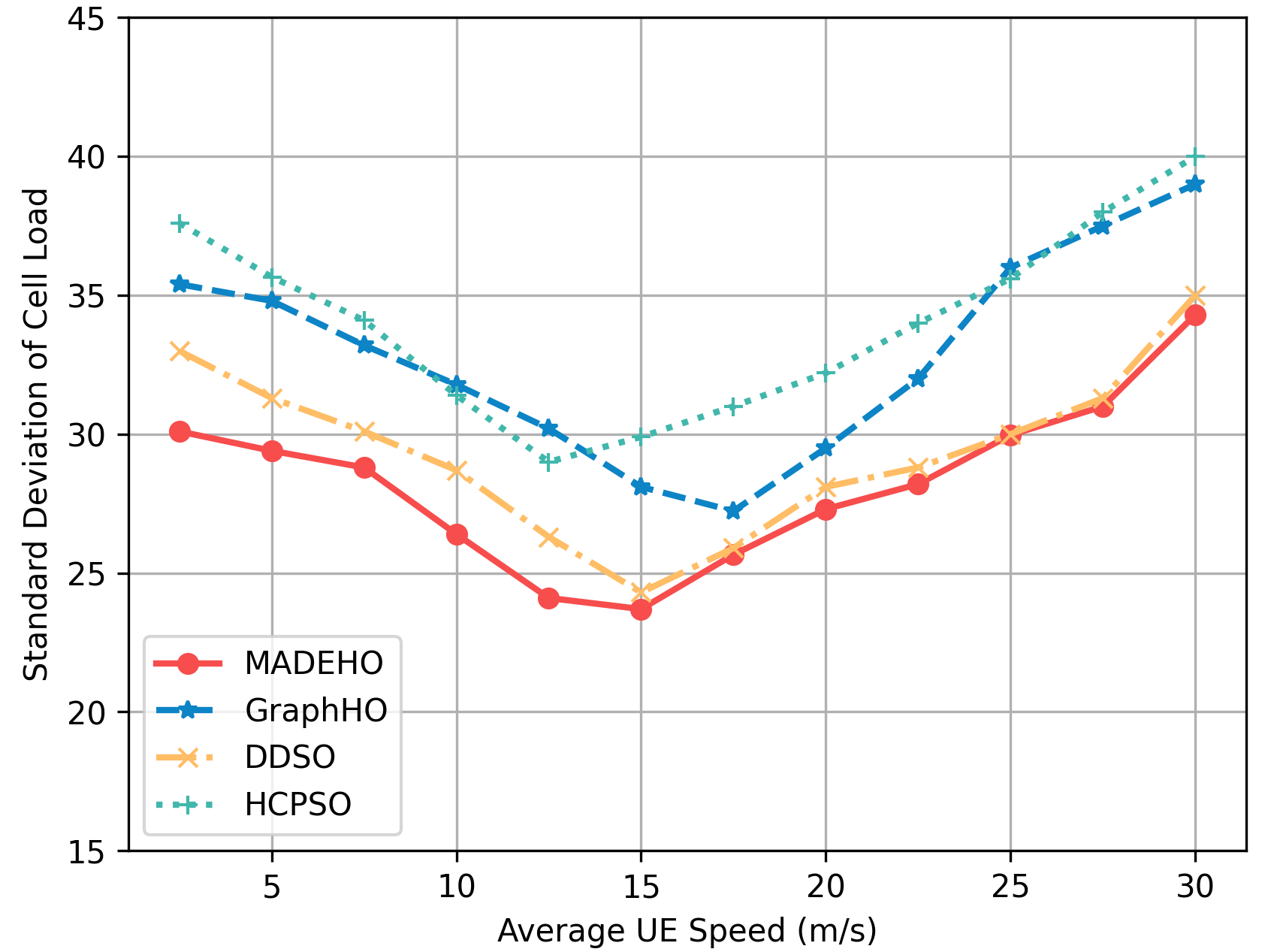}
\caption{Average UE speed vs. cell load std with the UE distribution std controlled within $[2.9, 3.1]$. }
\label{Fig.MovingSpeed}
\end{figure}

\begin{figure*}[!htb] 
\centering 
\includegraphics[width=\textwidth]{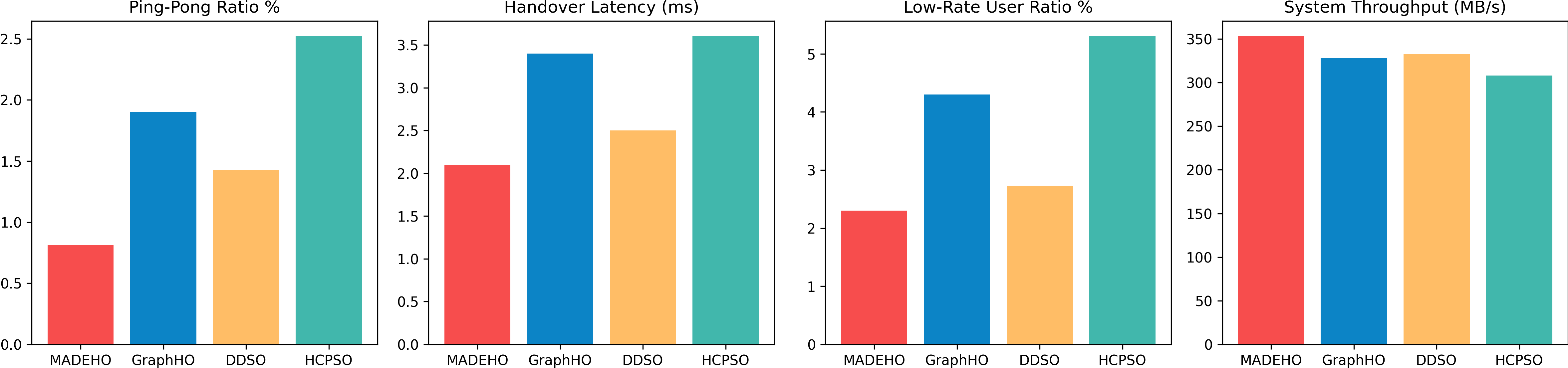}
\caption{Comparison of the four schemes in terms of ping-pong ratio, handover latency, system throughput, and low-rate user ratio.}
\label{Fig.bar}
\end{figure*}

\begin{figure}[htbp] 
\centering 
\includegraphics[width=0.45\textwidth]{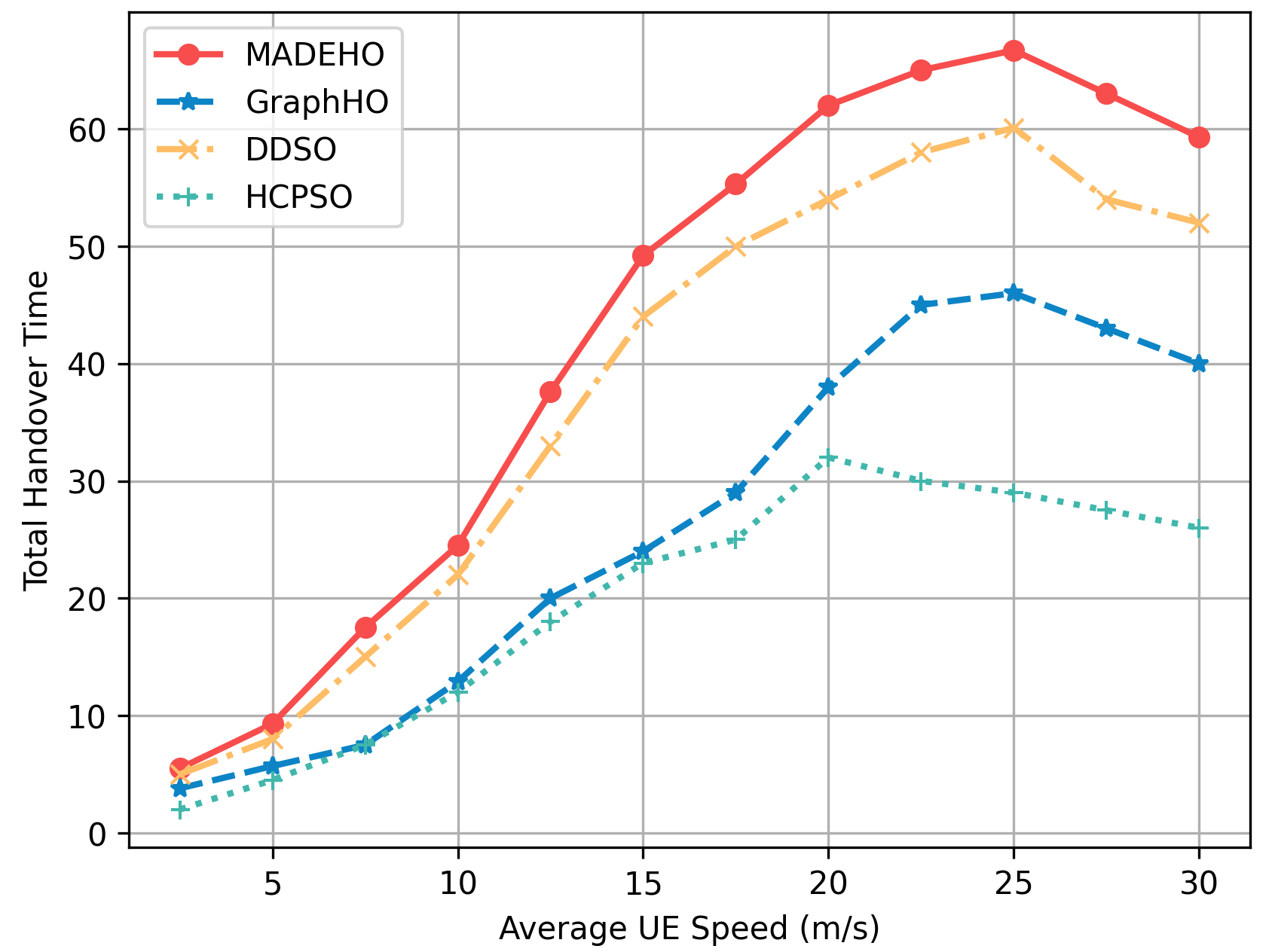}
\caption{Average UE speed vs. total handover time with the UE distribution std controlled within $[2.9, 3.1]$. }
\label{Fig.HandoverTime}
\end{figure}
\par To further illustrate the relationship between the average UE speed and load balancing difficulty, we show in Fig. \ref{Fig.HandoverTime} the curve of the total handover time within a period as the average UE speed changes, where the UE distribution std is controlled within $[2.9, 3.1]$. As the average UE speed increases, the handover time first increases since the UEs walk through different cells, and then decreases due to the unstable signal qualities. Handover is essential for achieving network load balancing. So the load balancing degree in Fig. \ref{Fig.MovingSpeed} first goes up and then falls down with the rise of the average UE speed.

\subsection{Handover and Network Performance Comparison}

While our primary objective is to achieve load balancing between cells, it is not the sole metric for evaluating cell handover and network performance. At the handover level, the ping-pong ratio (i.e., the switch-back ratio) and handover latency are also crucial considerations. Furthermore, at the network performance level, both system throughput and the low-rate UE ratio require attention. In this section, we define four performance metrics and compare the proposed scheme with baselines in terms of these performance metrics.

\begin{itemize}
  \item  \textbf{ping-pong ratio}: We define a ping-pong handover as a UE entering and exiting the same cell within a short period $T_{pp}$, which is similar to playing ping-pong. The ping-pong ratio refers to the proportion of ping-pong handovers in all handovers. Here we set $T_{pp}$ as $5$ time slots. A high ping-pong ratio indicates that the boundaries between cells are unclear or they significantly overlap.
  \item \textbf{handover latency}: The handover latency refers to the time required for a UE to complete the handover from the moment it starts listening in an inter-frequency handover. A long handover latency indicates that the UE takes a considerable amount of time to find a suitable cell, which suggests that the reporting parameters are too high or that cells are not sufficiently deployed.  
  \item \textbf{system throughput}: The system throughput equals to the total amount of data received by all UEs within a period, which reflects the overall network capacity.
  \item \textbf{low-rate user ratio}: We use this metric to represent the average proportion of UEs whose downlink rate is lower than a certain threshold $R_{low}$ during each time slot within a given period. Here we set $R_{low}$ as $1 \, \text{MB/s}$. The metric places great emphasis on user fairness.

\end{itemize}

\par Fig. \ref{Fig.bar} shows the performance of the four schemes in terms of the four metrics. The ping-pong ratio of our proposed scheme is $0.78\%$, which is the lowest among the four schemes. In terms of the handover latency, the four schemes reach $2.10 \, \text{ms}$,  $3.42 \, \text{ms}$,  $2.48 \, \text{ms}$,  and $3.55 \, \text{ms}$, respectively. As for the low-rate user ratio, our scheme achieves $2.21\%$, which is also lower than others. The system throughput of our scheme is $353 \, \text{MB/s}$, which is higher than those of the others by $7.62\%$, $6.0\%$, and $14.61\%$. The results demonstrate that our scheme can ensure network capacity and user fairness while implementing load balancing.

\subsection{Impact of Cell Frequency Distribution on Network Performance}
\par The previous experimental results are based on a given cell frequency distribution. However, different cell frequency distributions may lead to various network performance. In this part, we investigate the impact of adjusting the cell frequency distribution on load balancing and system throughput. Specifically, we modify the frequencies of some cells by adjusting certain 2.6 GHz and 4.9 GHz cells to 700 MHz without changing their position. Then we provide $1000$ episodes to test the modified environments, where the UE distribution std ranges from [2.9, 3.1], and the average UE speed is constrained between $24\,\text{m/s}$ to $26\,\text{m/s}$. In the modified environment, the ratios of intra-frequency and inter-frequency handovers change. We use the variable \textbf{intra-frequency neighboring cell ratio} to reflect this change, which is defined as the average proportion of intra-frequency neighboring cells of each cell.
\begin{figure}[htbp] 
\centering 
\includegraphics[width=0.45\textwidth]{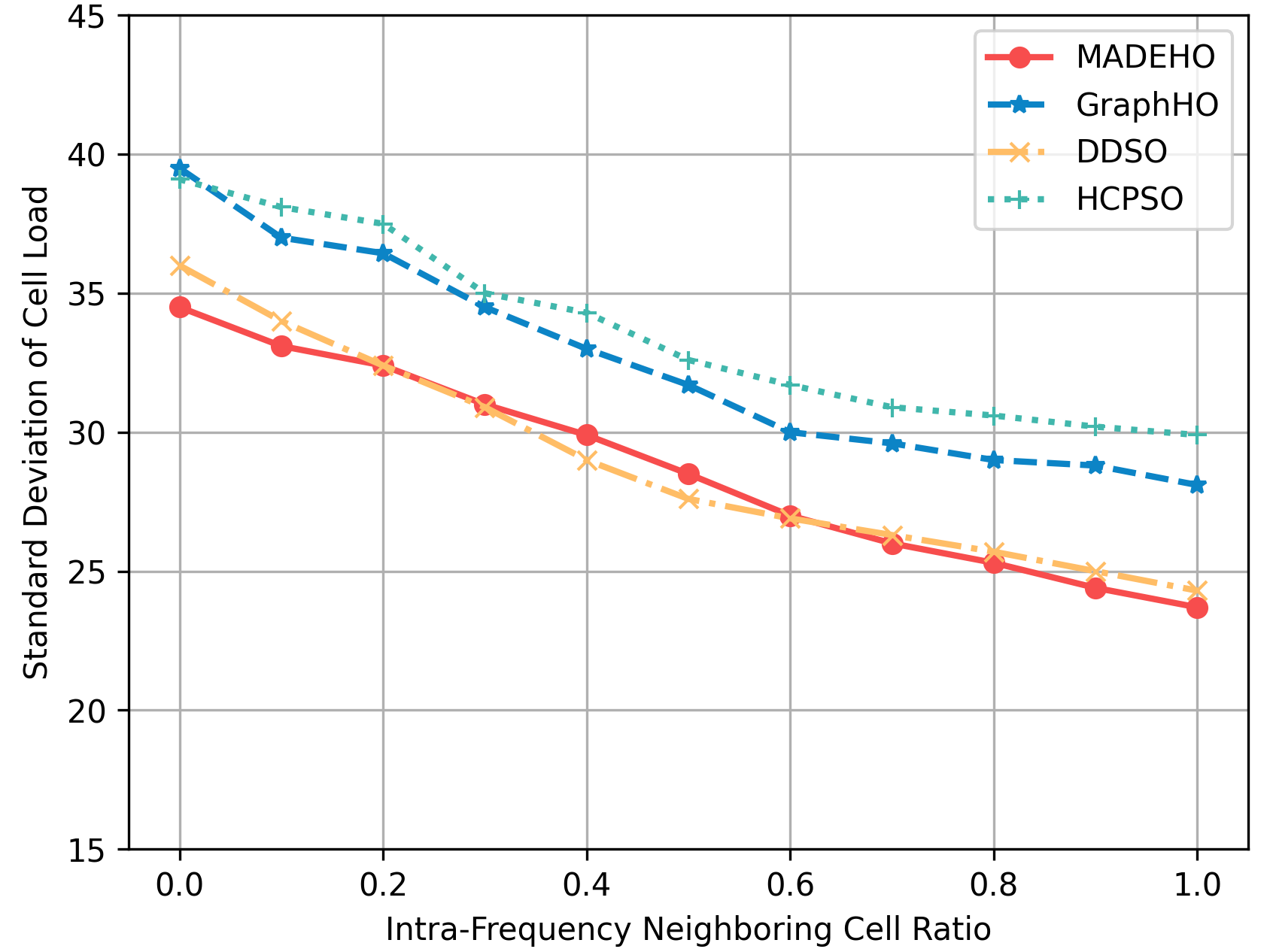}
\caption{Intra-frequency neighboring cell ratio vs. cell load std with the average UE speed in $[24, 26] \,\text{m/s}$ and the UE distribution std in $[2.9, 3.1]$. }
\label{Fig.Intraload}
\end{figure}

\par Fig. \ref{Fig.Intraload} shows that the cell load std falls down with the growth of the intra-frequency neighboring cell ratio. This is because intra-frequency handovers don't require monitoring and take less time, which makes it easier for UEs to switch to other cells, thus achieving load balancing. Our proposed scheme performs well in balancing load.

\begin{figure}[htbp] 
\centering 
\includegraphics[width=0.45\textwidth]{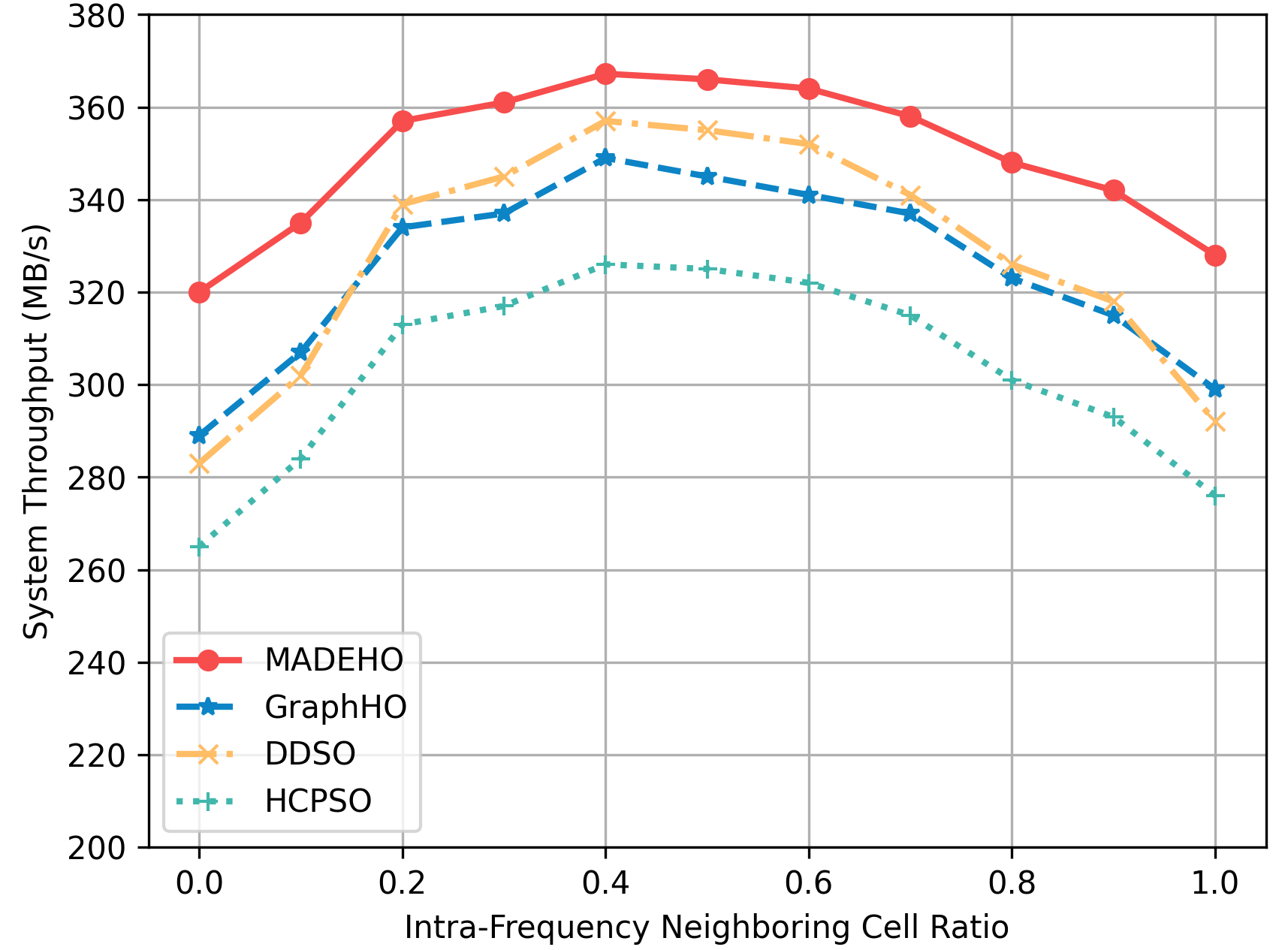}
\caption{Intra-frequency Neighboring cell ratio vs. system throughput with the average UE speed in $[24, 26] \,\text{m/s}$ and the UE distribution std in $[2.9, 3.1]$.}
\label{Fig.IntraThroughput}
\end{figure}

\par As shown in Fig. \ref{Fig.IntraThroughput}, the system throughput first rises up then turns down with a growing intra-frequency neighboring cell ratio. The reason is that a large intra-frequency neighboring cell ratio helps to reduce handover time and achieve load balancing. However, when it becomes too large, the interference between cells significantly increases, thereby reducing the system throughput. It is evident that our scheme outperforms other benchmarks.

\section{Conclusion}\label{conclusion}
In order to achieve load balancing in 5G cellular networks, in this paper we propose an MARL-based decentralized handover parameter optimization scheme. Firstly, we categorize cell handover into three types, model them following a four-step process, and define their handover conditions separately. Then we formulate the problem as minimizing the load standard deviation. This is a joint optimization problem that is difficult to solve directly. To promote cell cooperation, we solve it using an MARL method where decentralized training is implemented to reduce the cost of global real-time interaction. In detail, we replace the average load with the dynamic average consensus approximation based on local communication and prove that the error is bounded by a constant. Experimental results show that our proposed scheme significantly outperforms others in achieving load balancing and improving network performance.

\appendix[Proof of the error bound]\label{appendix}
\begin{proof}
The error between the estimate $\rho_{m,t}$ and the average load $\bar{L}_t$ is
\begin{equation}
    e_{m,t} = \rho_{m,t} - \bar{L}_t.
\end{equation}
We use $\boldsymbol{\rho_t} = [\rho_{1,t}, \dots, \rho_{M,t}]^T$ as the estimate vector at time slot $t$ and $\boldsymbol{e_t} = [e_{1,t}, \dots, e_{M,t}]^T$ as the error vector at time slot $t$. The update process (\ref{update}) can be rewritten as 
\begin{equation} \label{update2}
    \rho_{m,t+1}=\dot{L}_{m,t}+\sum_{j \in \textrm{N}_m}\frac{1}{\lvert\textrm{N}_m \lvert}\rho_{j,t}.
\end{equation}
We introduce $\rho_{m,t+1}=e_{m,t+1}+\bar{L}_{t+1}$ and $\rho_{j,t}=e_{j,t}+\bar{L}_t$ into (\ref{update2}), then we have
\begin{equation}\label{update3}
    e_{m,t+1}=\dot{L}_{m,t}-\dot{\bar{L}}_t+\sum_{j \in \textrm{N}_m}\frac{1}{\lvert\textrm{N}_m \lvert}e_{j,t},
\end{equation}
where $\dot{\bar{L}}_t=\bar{L}_{t+1}-\bar{L}_t$. We define $\boldsymbol{d}_L=\boldsymbol{\dot{L}_t}-\dot{\bar{L}}_t\boldsymbol{1}$ as the dynamic disturbance vector, $\boldsymbol{\dot{L}_t}= [\dot{L}_{1,t},\dots, \dot{L}_{M,t}]^T$. Since $\sup|\dot{L}_{m,t}| \leq \zeta < \infty$, we can deduce that 
\begin{equation}
    \|\boldsymbol{d}_L\|_\infty \leq  \|\boldsymbol{\dot{L}_t}\|_\infty+ \|\dot{\bar{L}}_t\boldsymbol{1}\|_\infty \leq 2\zeta.
\end{equation}
Then we can aggregate (\ref{update3}) for each $m$ into a vector form:
\begin{equation}
  \boldsymbol{e_{t+1}} =\boldsymbol{d}_L+\boldsymbol{\omega}\boldsymbol{e_{t}},
\end{equation}
where $\boldsymbol{\omega}$ is the consensus matrix related to the graph $G$. $\boldsymbol{\omega}_{m,j}=\frac{1}{\lvert\textrm{N}_m \lvert}$ if $j \in \textrm{N}_m $ and $\boldsymbol{\omega}_{m,j}=0$ otherwise. Considering that $\boldsymbol{\omega}$ is a row stochastic matrix where the sum of all elements in each row is $1$, and the graph \( G \) is connected and irreducible, the spectral radius of $\boldsymbol{\omega}$ in this error subspace $\boldsymbol{e_t}$ is less than 1 according to Perron-Frobenius theorem \cite{Perron-Frobenius}. Therefore, there exists a spectral radius $\lambda \in (0,1)$, such that for all $t$:

\begin{equation}
    \| \mathbf{e}_{t+1} \|_\infty \leq \lambda \| \mathbf{e}_t \|_\infty + 2 \zeta.
\end{equation}
This is a linear recursive relation whose solution is:
\begin{equation}
   \| \mathbf{e}_t \|_\infty \leq \lambda^{t-1} \| \mathbf{e}_1 \|_\infty + \frac{2 \zeta}{1 - \lambda}. 
\end{equation}
Since $\sup|L_{m,t}| \leq \zeta < \infty$ and $L_{m,t}>0$, the initial error satisfies $ \| \mathbf{e}_1 \|_\infty \leq \zeta $, we have:
\begin{equation}
    \| \mathbf{e}_t \|_\infty \leq  \zeta \lambda^{t-1} + \frac{2 \zeta}{1 - \lambda} \leq \frac{3-\lambda}{1-\lambda}\zeta
\end{equation}
As $ t $ increases, the term $  \zeta \lambda^{t-1} $ approaches zero, so we have:
\begin{equation}
    \| \mathbf{e}_t \|_\infty \leq \frac{2 \zeta}{1 - \lambda}.
\end{equation}
Therefore, for any cell $m$ and time $t$, the error satisfies:
\begin{equation}
    |e_{m,t}| \leq \| \mathbf{e}_t \|_\infty \leq \frac{3-\lambda}{1-\lambda}\zeta,
\end{equation}
where $ \lambda $ is the spectral radius of $ \omega$ in the error subspace $\boldsymbol{e_t}$, satisfying $0 < \lambda < 1$.
When $t \to \infty$ and each $\rho_{m,t}$ updates to a steady state, the error is by a smaller constant:
\begin{equation}
    \lim_{t \to \infty}|e_{m,t}| \leq \lim_{t \to \infty}\| \mathbf{e}_t \|_\infty \leq \frac{2\zeta}{1-\lambda}.
\end{equation}

\end{proof}

\begin{spacing}{1.1}
\bibliographystyle{IEEEtran}  
\bibliography{ref}     
\end{spacing}
\end{document}